\documentclass[a4paper, 12pt]{article}

\usepackage[english]{babel}
\usepackage[utf8]{inputenc}
\usepackage{amsmath}
\usepackage{amssymb}
\usepackage{graphicx}
\usepackage[colorinlistoftodos]{todonotes}
\usepackage{setspace}
\usepackage{pgfplots}

\usepackage{mathtools}
\pgfplotsset{compat=1.10}
\usepackage{tikz}
\usepackage{comment}
\usepackage[labelfont=bf]{caption}
\usepackage{titlesec}
\mathtoolsset{showonlyrefs=true}
\usepackage{csquotes}
\usepackage{bbm}
\usepackage{dcolumn,booktabs}
\usetikzlibrary{snakes}
\usetikzlibrary{patterns,decorations.pathreplacing}
\usepackage[top=1in, bottom=1in, left=1in, right=1in]{geometry}
\usepackage{bm}

\makeatletter
\renewcommand\@biblabel[1]{}
\makeatother
\usepackage{rotating}
\usepackage{amsthm}
\usepackage{thmtools, thm-restate}
\usepackage{environ}

\usepackage{setspace}
\usepackage{authblk}
\usepackage{hyperref}

\usepackage{pdflscape}
\usepackage{subfig}
\usepackage{natbib}
\usepackage{enumitem}
\usepackage{bbm}

\newtheorem{claim}{Claim}

\newtheorem{lemma}{Lemma}

\newtheorem{assumption}{Assumption}

\newtheorem{definition}{Definition}

\AfterEndEnvironment{prop}{\noindent\ignorespaces}
\AfterEndEnvironment{example}{\noindent\ignorespaces}
\AfterEndEnvironment{theorem}{\noindent\ignorespaces}
\AfterEndEnvironment{lemma}{\noindent\ignorespaces}
\AfterEndEnvironment{result}{\noindent\ignorespaces}
\AfterEndEnvironment{assumption}{\noindent\ignorespaces}
\AfterEndEnvironment{corollary}{\noindent\ignorespaces}
\AfterEndEnvironment{definition}{\noindent\ignorespaces}
\AfterEndEnvironment{claim}{\noindent\ignorespaces}
\AfterEndEnvironment{app}{\noindent\ignorespaces}
\AfterEndEnvironment{sublemma}{\noindent\ignorespaces}

\usepackage{mdframed}
\mdfdefinestyle{myenvs}{%
  hidealllines=true,%
  nobreak=true, % comment this to allow breaking
  leftmargin=0pt,
  rightmargin=0pt,
  innerleftmargin=0pt,
  innerrightmargin=0pt,
}

\newmdtheoremenv[style=myenvs]{prop}{Proposition}

\newcolumntype{L}[1]{>{\raggedright\let\newline\\\arraybackslash\hspace{0pt}}m{#1}}
\newcolumntype{C}[1]{>{\centering\let\newline\\\arraybackslash\hspace{0pt}}m{#1}}
\newcolumntype{R}[1]{>{\raggedleft\let\newline\\\arraybackslash\hspace{0pt}}m{#1}}

\newlist{enumdescript}{description}{1}
\setlist[enumdescript,1]{%
  before={\setcounter{descriptcount}{0}%
          \renewcommand*\thedescriptcount{\arabic{descriptcount}}},
        font={\bfseries\stepcounter{descriptcount}Application \thedescriptcount:~}
}
\newcounter{descriptcount}

\hypersetup{colorlinks=true, linkcolor=black, anchorcolor=black, citecolor=black, filecolor=black, menucolor=black, runcolor=black, urlcolor=blue}

\usepackage{everypage}

\newlength{\hfoot}
\newlength{\vfoot}
\AddEverypageHook{\ifdim\textwidth=\linewidth\relax
\else\setlength{\hfoot}{-\topmargin}%
\addtolength{\hfoot}{-\headheight}%
\addtolength{\hfoot}{-\headsep}%
\addtolength{\hfoot}{-.5\linewidth}%
\ifodd\value{page}\setlength{\vfoot}{\oddsidemargin}%
\else\setlength{\vfoot}{\evensidemargin}\fi%
\addtolength{\vfoot}{\textheight}%
\addtolength{\vfoot}{\footskip}%
\raisebox{\hfoot}[0pt][0pt]{\rlap{\hspace{\vfoot}\rotatebox[origin=cB]{90}{\thepage}}}\fi}


\title{\singlespacing{\textbf{Reputational Algorithm Aversion}}\thanks{McGill University. Email: gregory.weitzner@mcgill.ca. I thank Jeremy Bertomeu, Adolfo De Motta, Vincent Glode and seminar participants at the Columbia/RFS AI in Finance Conference for the helpful comments and discussions.}}

\author{Gregory Weitzner}

\date{July 2024}

\begin{document}
\maketitle

\begin{abstract}
\begin{small}
\begin{singlespace}
\noindent People are often reluctant to incorporate information produced by algorithms into their decisions, a phenomenon called ``algorithm aversion''. This paper shows how algorithm aversion arises when following or overriding an algorithm conveys information about a human's ability. In my model, workers make forecasts of an uncertain outcome based on their private information and an algorithm's signal. Although it is efficient for low-skill workers to follow the algorithm, they sometimes override it due to reputational concerns. The model provides a fully rational microfoundation for algorithm aversion that aligns with the broad concern that AI systems will displace many types of workers.

\end{singlespace}
\end{small}
\end{abstract}

\clearpage

\section{Introduction}

In recent years, algorithms have become better at forecasting many types of outcomes than humans, raising concerns that artificial intelligence will cause the mass displacement of jobs.\footnote{I follow \cite{ludwig2021fragile} by using the phrase algorithm interchangably with ``artificial intelligence'', ``machine learning'' or ``deep learning''. An incomplete list of papers that show algorithms outperform humans across many domains include: \cite{yeomans2019making}, \cite{lai2021towards}, \cite{mullainathan2019machine}, \cite{liu2022assessing} and \cite{agarwal2023combining}. } However, many argue that these concerns are overstated: instead of being displaced by algorithms, humans will play a critical role in collaborating with and/or supervising algorithms.\footnote{E.g., \cite{langlotz2019will}, \cite{agrawal2019artificial}, \cite{ludwig2021fragile} \cite{acemoglu2023can}. For example, \cite{ludwig2021fragile} argue that in many contexts, humans can supplement algorithm's because they have better context/can better assess rare situations by using intuition.} Based on this idea, a rapidly growing literature across computer science, economics and psychology explores how humans work alongside artificial intelligence systems.

However, despite the clear value of algorithms, humans are often reluctant to incorporate algorithmic forecasts into their own decisions, a phenomenon which \cite{dietvorst2015algorithm} coin ``algorithm aversion''. Algorithm aversion has been documented empirically in domains as disparate as health (e.g., \cite{promberger2006patients}, \cite{longoni2019resistance}, \cite{shaffer2013patients} and \cite{agarwal2023combining}), finance (e.g., \cite{onkal2009relative}, \cite{niszczota2020robo} and \cite{greig2023algorithm}) and judicial decision making (e.g., \cite{angelova2023algorithmic}). In some cases, even after being provided the algorithm's forecast, humans' forecasts perform worse than the algorithm on its own. At first blush, this behavior is surprising. If humans act as rational Bayesians, their forecasts should always be at least as accurate as an algorithm's. For this reason, the main explanations for algorithm aversion center on behavioral biases and psychological reasons.\footnote{See \cite{burton2020systematic}, \cite{jussupow2020we} and \cite{mahmud2022influences} for review articles discussing the various explanations for algorithm aversion.} Given the large potential productivity benefits of AI-human collaboration, it is critical to understand why exactly humans are so reluctant to incorporate valuable information from algorithms into their own decision making. Moreover, understanding the reasons for algorithm aversion can also help in design more efficient systems in which humans work alongside algorithms. 

In this paper, I argue that a simple, yet overlooked, reason for algorithm aversion is that a human's choice to follow or override an algorithm's forecast may convey information about that human's ability. To illustrate this idea, consider a situation in there is worker forecasting a binary outcome with the help of an algorithm provided by the firm. After observing the algorithm's forecast, the worker makes his own forecast and the outcome is realized. The worker has private information about his skill, where he can either be high-skill, in which case he is better at predicting the outcome than the algorithm, or he can be low-skill, in which case he is worse than the algorithm. Hence, it is efficient for the low-skill worker to report the algorithm's forecast, while the high-skill worker report his own information. This paper shows that the efficient use of the algorithm can never be achieved. Why? If the algorithm is used efficiently, only the high-skill worker ever overrides the algorithm. However, if this were the case, the worker would always override the algorithm to convince the firm they are high-skill. Hence, in equilibrium the low-skill worker must at least sometimes inefficiently override the algorithm. Put in more simple terms, in order for a worker to convince his boss that he is worth keeping around, he needs to sometimes override the algorithm, even if doing so is the wrong decision for the company.

Formally, I analyze a strategic communications, i.e., cheap talk, game in which a worker is tasked with predicted an uncertain outcome. The worker receives a private signal as well as a public signal produced by an algorithm. After receiving both signals, the worker reports a forecast, after which the outcome is realized. Importantly, the worker has private information about his own skill. Specifically, the worker can either be low-skill, in which case his information is always inferior to the algorithm, or he can be high-skill in which case his information is always superior to the algorithm. The firm prefers the worker report an accurate forecast;  hence, from the firm's perspective, it is efficient for the low-skill worker to report the algorithm's forecast, while the high-skill worker always reports his own signal. After receiving the worker's forecast and the realization of the outcome, the firm updates its beliefs regarding the worker's skill. Importantly, the worker faces reputational concerns (e.g., \cite{holmstrom1999managerial}) and hence reports a forecast which maximizes the firm's perception that he is high-skill. 

I first consider a benchmark case with no algorithm. In this benchmark, there always exists an efficient equilibrium in which the worker honestly reports his signal. Intuitively, the firm assesses the worker entirely on the accuracy of his forecast which incentivizes the worker to report his signal truthfully in order for his forecast to be perceived as accurate as possible. In the main model in which the worker has access to the algorithm's signal, this is no longer the case. Here it is efficient for the low-skill worker to report the algorithm's forecast, while the high-skill worker should always report his own signal. However, whereas in the benchmark case the forecast on its own provided no information regarding the worker's skill, here the firm learns both from the accuracy of the worker's forecast and the forecast itself when compared to the algorithm's signal. In fact, because of this latter effect the first-best use of information can never be achieved. The reason behind this result is simple: if the high-skill worker is the only type to ever override the algorithm, overriding the algorithm will instantly reveal that the worker is high-skill. Hence, under this set of beliefs, both worker types would always deviate by reporting the opposite signal of the algorithm. 

I next show that there exists an informative equilibrium in which the high-skill worker always reports his signal and the low-skill worker reports his signal with some positive probability when his signal differs from the algorithm. In this equilibrium, the worker exhibits ``algorithm aversion'', i.e., he overrides the algorithm even when his own signal is less accurate than the algorithm. Hence, algorithm aversion causes the overall accuracy of the worker's forecast to be strictly lower than the first-best. In fact, in some cases the accuracy of the worker's forecast is worse than the algorithm alone and the worker provides at best zero static benefits to the firm. 

I next show that algorithm aversion increases in the uncertainty of the algorithm. Intuitively, as the algorithm becomes more uncertain it becomes less costly to override it because its signal is less accurate. However, there is also an interesting second effect. Because the algorithm is less accurate, and the high-skill worker always reports his own signal, the high-skill worker's signal differs from the algorithm more often and hence he defies the algorithm more frequently in equilibrium. This increases the signaling value of overriding the algorithm for the low-skill worker. Hence, this indirect effect causes an even larger increase in algorithm aversion. 

Beyond generating a rational microfoundation for algorithm aversion, the model's predictions are consistent with substantial empirical evidence. First, humans' forecasts are often less accurate than an algorithm's even after being given access to the algorithm's forecast (e.g., \cite{angelova2023algorithmic} and \cite{agarwal2023combining}), which can occur enogenously in my model for completely rational reputational reasons. Relatedly, \cite{angelova2023algorithmic} analyze judges' discretion over algorithm and find that 90\% of judges underperform an algorithm when they override it's recommendation. In my model, the worker overrides the algorithm too frequently and the accuracy of his forecast in the case of an override is often lower than that of the algorithm. While part of \cite{angelova2023algorithmic}'s results could be due to judges simply making mistakes, my model shows how override decisions can also be driven by reputational concerns. If a worker never override an algorithm, it can be difficult for them to justify keeping their job.

The main mechanism of my model is that by relying on an algorithmic recommendation, a human conveys negative information about their ability. Consistent with this prediction, \cite{shaffer2013patients} show that patients rate doctors using computerized recommendations lower than those that do not rely on computerized recommendations.\footnote{\cite{leyer2019me} show that humans are less likely to delegate decisions to an AI system versus another human. Although there is no delegation decision in my model, if such a delegation decision reflects negatively on the human, my model would also predict a reluctance to do so.} While data about the use of algorithms within firms is limited, the model would predict that the same type of dynamic would exist for workers relying on algorithmic forecasts within firms. 

To my knowledge this is the first model providing a mechanism for algorithm aversion which is consistent with rational Bayesian behavior. A non-exhaustive list of other explanations include i) the desire to understand forecasts (e.g., \cite{yeomans2019making}), ii) valuing agency of decisions (e.g., \cite{sunstein2023use}),  iii) humans having incorrect priors (e.g., \cite{yeomans2019making}, \cite{greig2023algorithm}) and iv) humans deriving utility from interacting with another human and/or disutility from interacting with a machine
(e.g., \cite{dietvorst2015algorithm} and \cite{greig2023algorithm}).\footnote{See \cite{jussupow2020we}, \cite{burton2020systematic} and \cite{mahmud2022influences} for review articles of algorithm aversion.} While I do not doubt that psychological and behavioral reasons may partially explain this phenomenon, I argue that a simple, yet missing part of the story is the reputational concerns of the people using these algorithms. More broadly, by highlighting an impediment of humans working alongside AI systems, this paper contributes to the literature on human-AI collaboration.\footnote{See \cite{lai2021towards} for an excellent review of this literature.} 

In the model, algorithm aversion makes the worker's forecast less accurate and in some cases the worker provides zero value to the firm in the short-run. In practice, this can put downward pressure on wages for humans. However, if firms are are impatient and it is costly to retain/train workers, reputational algorithm aversion could result in firms firing workers entirely. For these reasons, this paper relates to the literature analyzing the use of AI within firms and in AI's affect on labor market outcomes (e.g., \cite{acemoglu2018artificial}, \cite{acemoglu2020wrong}, \cite{acemoglu2022artificial},  \cite{cao2023talk}, \cite{babina2024artificial}, \cite{jha2024chatgpt} and \cite{eisfeldt2023generative}).

Finally, from a theory perspective, the model follows the reputational cheap-talk framework (e.g., \cite{scharfstein1990herd}, \cite{ottaviani2006professional} \cite{ottaviani2006reputational},  \cite{ottaviani2006strategy}, \cite{guembel2009reputational}) which combines cheap talk (\citealp{crawford1982strategic}) with career concerns (\citealp{holmstrom1999managerial}). Like other reputational cheap talk models, agents make strategic forecasts in order to maximize the likelihood that they are perceived as skilled. However, there are two main differences between my model and typical reputational cheap-talk models. First, in my model the worker knows his type, whereas typically the sender's type is unknown to himself. Second, the worker has access to both a private signal and a public signal (i.e., the algorithm). These two ingredients cause the worker's forecast itself to convey information about his ability beyond the accuracy of the forecast, which results in workers placing too little weight on the public signal due to reputational concerns. 

\section{Model Setup}
There is a worker who is tasked by his manager to predict an outcome regarding a state of the world $\omega \in \{\omega_0,\omega_1\}$. The worker is of type $\theta \in \{\theta_L,\theta_H\}$ (low-skill or high-skill) and receives a private signal $s \in \{s_0,s_1\}$ regarding $\omega$, where both the worker's type and signal are only known by the worker. The worker is also given access to an algorithm developed by the firm, which produces a signal $a \in \{a_0,a_1\}$ regarding $\omega$, where the algorithm's signal is observable to both the worker and the manager. I assume the algorithm is a ``black-box'' which obtains information regarding $\omega$ by analyzing a set of training data which I do not explicitly model.  

The worker observes both signals then reports a message $m \in \{m_0,m_1\}$, i.e., forecast, to the manager, where $m = m_0$ means (in equilibrium) that the worker predicts that $\omega = \omega_0$.\footnote{Because this is a cheap-talk model, the messages only have meaning to the extent they convey information to the manager in equilibrium.} After the worker reports his message the state $\omega$ is realized and observed by all. 

Because the worker's actions are costless, the model applies to any situation in which the worker's action serves as a forecast, potentially conveying information about his skill. For example, it could be the worker is forecasting the firm's earnings or cash flows, but could also be be that the worker is making an investment decision (e.g., \cite{scharfstein1990herd}).\footnote{In \cite{scharfstein1990herd} the outcome of the investment decision is observable whether or not the firm invests. Hence, this interpretation would be most relevant to situations in which the the firm can tell whether the action would have succeeded regardless of whether the action was taken, e.g., if the investment outcome is mostly tied to market-wide conditions. }

\subsection{Information Structure}
The worker is equally likely to be low-skill or high-skill, i.e., $Pr(\theta_H) = Pr(\theta_L) = \frac{1}{2}$ and both states of the world are equally likely, $Pr(\omega_1) = Pr(\omega_0) = \frac{1}{2}$. When $\omega = \omega_1$, the high-skill worker receives the signal $s_1$ with probability $\upsilon_H$, while the low-skill worker receives it with probability $\upsilon_L < \upsilon_H$. When $\omega = \omega_0$, the high-skill worker receives the signal $s_1$ with probability $1-\upsilon_H$ and the low-skill worker receives it with probability $1-\upsilon_L$:
\begin{align}
Pr(s_1|\omega_1, \theta_L) &\equiv \upsilon_L > \frac{1}{2}, \\ 
Pr(s_1|\omega_1, \theta_H) &\equiv \upsilon_H > \upsilon_L, \\
Pr(s_1|\omega_0, \theta_L) &\equiv 1-\upsilon_L, \\
Pr(s_1|\omega_0, \theta_H) &\equiv 1-\upsilon_H.
\end{align}
These assumptions make it such that the high-skill worker always has more accurate information than the low-skill worker and that the ex-ante distribution of signals is the same for both low-skill and high-skill workers, i.e.,  $Pr(s_1| \theta_H) =  Pr(s_1| \theta_L) = \frac{1}{2}$. The latter simplifies the problem by ensuring that the signal itself does not convey any information about the skill of the worker.\footnote{A similar assumption is made in \cite{scharfstein1990herd}.} 

When $\omega = \omega_1$, the algorithm produces the signal $a_1$ from the algorithm with probability $\alpha$, while if $\omega = \omega_0$ the algorithm produces signal $a_1$ with probability $1-\alpha$. Importantly, the algorithm's recommendation does not depend on the worker's skill. This assumption is reasonable so long as the algorithm is developed by other employees of the firm or outside the firm entirely. In practice, different firms may even use the same algorithm developed by a third party.\footnote{For example,  firms often use the same large language models, e.g., ChatGPT.} 

To make the problem interesting, I assume that the high-skill worker's signal is more informative than the algorithm, while the low-skill worker's is less informative.
\begin{assumption}\label{a1}
the high-skill worker's signal is more informative than the algorithm, while the low-skill worker's is less informative, i.e., $\alpha \in (\upsilon_L, \upsilon_H)$.
\end{assumption}
Without this assumption, it would always be optimal for the worker to either always report the algorithm's signal or always report his own signal. In practice, the ``ground truth'' in algorithms' training data is often established based on ``expert'' assessments. Similarly, \cite{angelova2023algorithmic} find that 10\% of judges outperform an algorithm in the context of bail decisions. Hence, it is natural that at least some workers can forecast better than the algorithm. Moreover, if the algorithm always provided the better forecast, there would be no need to use the worker's information at all.

\subsection{Worker's Objective and Equilibrium Concept}
I assume that the manager always prefers that the worker's forecast is correct but do not explicitly model why this is valuable to the manager.\footnote{This is standard in reputational cheap talk models (e.g., \cite{scharfstein1990herd}, \cite{ottaviani2006professional}, \cite{ottaviani2006reputational}, \cite{ottaviani2006strategy}.} For example, the firm may prefer that its treasurer make more accurate cash flow estimates to help guide firms' financial policies. Relatedly, the firm would prefer to make investments when market conditions turn out to be favorable for that investment. 

The manager first observes the worker's forecast $m$ and the algorithm's signal $a$, after which the outcome $\omega$ is realized. The manager then forms a posterior belief $\hat{\theta}(m,a,\omega)$ on the probability that the worker is high-skill.\footnote{We can equivalently think of the manager as anyone who observes both the worker and the algorithm's recommendation. If the worker and algorithm's forecast are observable to other firms, then we think of the manager as the entire labor market (e.g., \cite{holmstrom1999managerial}).}

Through assessing the worker's forecast and the accuracy of the forecast in predicting the outcome, the manager can learn about the worker's skill. As is typical in career concerns models, there are no long-term contracts and hence, the worker attempts to maximize the manager's perceived probability that the worker is high-skill. This is equivalent to a risk-neutral worker maximizing his next period, competitive spot wage, where the wage is linear in the manager's posterior belief that the worker is high-skill.\footnote{This approach is common in reputational cheap talk models (e.g., \cite{scharfstein1990herd}).} Hence, the worker's payoff following his own signal $s$, the algorithm's signal $s$ and the worker's message $m$ as follows:
\begin{gather}
     \sum_\omega   Pr(\omega | s, \theta)  \hat{\theta}(m, a, \omega).
\end{gather}
The equilibrium concept I use throughout the analysis is perfect Bayesian Nash Equilibrium, in which the worker's equilibrium strategy maximizes his expected reputational payoff and the manager's conditional probabilistic belief is correct on equilibrium path, i.e., $\hat{\theta}(m,a,\omega) = Pr(\theta_H|m,a,\omega)$ for $\omega \in \{\omega_0, \omega_1\}$, $a \in \{a_0, a_1 \}$ and any $m$ that is used in equilibrium.

\section{Benchmark: No Algorithm}
First, it will be useful to establish a benchmark in which the worker does not have access to the algorithm's signal to highlight the mechanism throughwhich algorithm aversion arises. Because the manager always prefers a more accurate forecast, it is efficient for the worker to report his signal truthfully. For this analysis, it will be useful to calculate the worker's posterior belief regarding the likelihood of state $\omega$:
\begin{align}
Pr(\omega_1 | s_1, \theta_H) = \frac{Pr(s_1 | \omega_1, \theta_H)Pr(\omega_1|\theta_H)}{Pr(s_1 | \theta_H)} &= \upsilon_H,\\
Pr(\omega_1 | s_0, \theta_H) = \frac{Pr(s_0 | \omega_1, \theta_H)Pr(\omega_1|\theta_H)}{Pr(s_0 | \theta_H)} &=1-\upsilon_H, \\
Pr(\omega_1 | s_1, \theta_L) = \frac{Pr(s_1 | \omega_1, \theta_L)Pr(\omega_1|\theta_L)}{Pr(s_1 | \theta_L)} &= \upsilon_L,\\
Pr(\omega_1 | s_0, \theta_L) = \frac{Pr(s_0 | \omega_1, \theta_L)Pr(\omega_1|\theta_L)}{Pr(s_0 | \theta_L)} &=1-\upsilon_L. 
\end{align}
It will also be useful to calculate the manager's posterior beliefs assuming that the worker engages in truth-telling:
\begin{align}
    \hat{\theta}(\omega_1, s_1) &= \frac{\upsilon_H}{\upsilon_H + \upsilon_L}, \\
    \hat{\theta}(\omega_1, s_0) &= \frac{1-\upsilon_H}{2-\upsilon_H - \upsilon_L}, \\
        \hat{\theta}(\omega_0, s_1) &= \frac{1-\upsilon_H}{2-\upsilon_H - \upsilon_L}, \\
    \hat{\theta}(\omega_0, s_0) &= \frac{\upsilon_H}{\upsilon_H + \upsilon_L}.
\end{align}
In order for the low-skill and high-skill workers to truthfully report $m_1$ when $s=s_1$, the following incentive compatibility conditions must hold: 
\begin{align}
  \sum_\omega   Pr(\omega | s_1, \theta_L)  \hat{\theta}(s_1, \omega) &\geq  \sum_\omega   Pr(\omega | s_1, \theta_L)  \hat{\theta}(s_0, \omega) \label{benchmark1} \\ 
  \sum_\omega   Pr(\omega | s_1, \theta_H)  \hat{\theta}(s_1, \omega) &\geq  \sum_\omega   Pr(\omega | s_1, \theta_H)  \hat{\theta}(s_0, \omega). \label{benchmark2}
\end{align}
I next show that in this benchmark case, there always exists an equilibrium in which the worker reports his signal as his forecast. 
\begin{prop}
In the benchmark case without the algorithm, a truth-telling equilibrium always exists in which both the low-skill worker and high-skill worker truthfully report their signal, i.e., $m = m_0$ if $s = s_0$ and $m = m_1$ if $s = s_1$.
\end{prop}

\begin{proof}
Because the problem is symmetric, the incentives to deviate will always be the same regardless of the signal the worker receives. Hence, it is without loss of generality to consider the incentives to deviate following signal $s_1$. Plugging in the posteriors calculated above into \eqref{benchmark1} and simplifying we have:
    \begin{align}
\frac{\left(\upsilon _H-\upsilon _L\right) \left(2 \upsilon
   _L-1\right)}{\left(2-\upsilon _H-\upsilon _L\right)
   \left(\upsilon _H+\upsilon _L\right)},
    \end{align}
    which is positive given that $\upsilon_H > \upsilon_L$ and $\upsilon_L > \frac{1}{2}$. Hence, the low-skill worker will not deviate. Plugging in the posteriors into \eqref{benchmark1} and simplifying we have:
        \begin{align}
\frac{\left(\upsilon _H-\upsilon _L\right) \left(2 \upsilon
   _H-1\right)}{\left(2-\upsilon _H-\upsilon _L\right)
   \left(\upsilon _H+\upsilon _L\right)},
    \end{align} 
    which is also clearly positive. Hence, the high-skill worker will also not deviate from truth-telling. 
\end{proof}
Intuitively, because the worker's forecast itself does not convey any information about the worker's ability, the worker can only convey his skill by making as accurate of a forecast as possible. Hence, it is incentive compatible for the worker to report his signal truthfully if the manager expects him too. However, once I introduce the algorithm, the worker's forecast will convey information about his skill because it can be compared to the algorithm's signal. 

It is worth mentioning that as in all cheap talk models, there also exists a ``babbling'' equilibrium, in which both high and low-skill workers report a forecast $m$ with the same probability regardless of their type and the signal they receive. In this case, the manager cannot glean any information from the worker's message and hence, the worker is indifferent between any potential message he sends.

\section{Algorithm}
Next, I analyze the case in which the worker has access to the algorithm before making his forecast. Because the problem is symmetric, throughout the analysis I focus on the case in which the worker receives signal $s_1$; however, the analysis is the same for the case in which $s = s_0$. Again, it will be useful to calculate the worker's posterior belief regarding the likelihood of the state:
\begin{align}
 Pr(\omega_1 | s_1, a_1, \theta_L) = \frac{Pr(s_1, a_1 | \omega_1, \theta_L)Pr(\omega_1| \theta_L)}{Pr(s_1, a_1 | \theta_L)} &= \frac{\alpha \upsilon_L}{\alpha \upsilon_L + (1-\alpha)(1-\upsilon_L)}, \label{post1}  \\ 
  Pr(\omega_1 | s_1, a_1, \theta_H) = \frac{Pr(s_1, a_1 | \omega_1, \theta_H)Pr(\omega_1| \theta_H)}{Pr(s_1, a_1 | \theta_H)} &= \frac{\alpha \upsilon_H}{\alpha \upsilon_H + (1-\alpha)(1-\upsilon_H)},  \label{post2} \\ 
  Pr(\omega_1 | s_1, a_0, \theta_L) = \frac{Pr(s_1, a_0 | \omega_1, \theta_L)Pr(\omega_1| \theta_L)}{Pr(s_1, a_0 | \theta_L)} &= \frac{(1-\alpha) \upsilon_L}{(1-\alpha) \upsilon_L + \alpha(1-\upsilon_L)},  \label{post3} \\
Pr(\omega_1 | s_1, a_0, \theta_H) = \frac{Pr(s_1, a_0 | \omega_1, \theta_H)Pr(\omega_1| \theta_H)}{Pr(s_1, a_0 | \theta_H)} &= \frac{(1-\alpha) \upsilon_H}{(1-\alpha) \upsilon_H + \alpha(1-\upsilon_H)}. \label{post4}
\end{align}
Based on these posterior probabilities, the manager would always prefer that the low-skill worker report the algorithm's forecast rather than his own signal, while the high-skill worker reports his own signal. To see this, first note that from Assumption \ref{a1}, $\upsilon_H > \alpha > \upsilon_L > \frac{1}{2}$. Hence, when the worker's signal coincides with that of the algorithm, the posterior probability that the outcome coincides with the worker's signal is always greater than one half, i.e., $Pr(\omega_1 | s_1, a_1, \theta) > \frac{1}{2}$ for $\theta \in \{\theta_L, \theta_H \}$. However, when the worker's signal differs from the algorithm, the high-skill worker's posterior probability that the outcome coincides with the signal is greater than one-half, i.e., $Pr(\omega_1 | s_1, a_0, \theta_H) > \frac{1}{2}$, while the low skill's worker is less than one-half, i.e., $Pr(\omega_1 | s_1, a_0, \theta_L) < \frac{1}{2}$. Hence, for the forecast to be as accurate as possible, the manager always prefers that the high-skill worker report his own signal, while the low-skill worker always report the algorithm's signal.  

Henceforth, I refer to the case in which the worker reports the algorithm's signal when he is low-skill and his own signal when he is high-skill the ``first-best''. I next show that because of reputational concerns, the first-best can never be achieved as an equilibrium. 
\begin{prop}\label{propfb}
    The first-best cannot be achieved as an equilibrium. 
\end{prop}

\begin{proof}
First, it is useful to calculate the following posterior beliefs given that the manager expects the worker to follow the first-best use of information: 
\begin{align}
    \hat{\theta}^{FB}( m_1, a_1, \omega_1) &= \frac{\upsilon_H}{\upsilon_H + 1},\\
    \hat{\theta}^{FB}( m_1, a_1, \omega_0) &= \frac{1-\upsilon_H}{1-\upsilon_H + 1},\\
    \hat{\theta}^{FB}( m_0, a_1, \omega_1) &= 1,\\
    \hat{\theta}^{FB}( m_1, a_1, \omega) &= 1.
\end{align}
After receiving signal $s_1$, incentive compatibility for the low-skill worker requires: 
\begin{align}\label{ica}
   \sum_\omega  Pr(\omega | s_1, a_1, \theta_L)  \hat{\theta}^{FB}( m_1, a_1, \omega) &\geq  \sum_\omega  Pr(\omega | s_1, a_1, \theta_L)  \hat{\theta}^{FB}( m_0, a_1, \omega), \\
      \sum_\omega  Pr(\omega | s_1, a_0, \theta_L)  \hat{\theta}^{FB}( m_0, a_0, \omega) &\geq  \sum_\omega  Pr(\omega | s_1, a_0, \theta_L)  \hat{\theta}^{FB}( m_1, a_0, \omega), 
\end{align}
where \eqref{ica} says that when the low-skill worker receives the signal $s_1$ he should report $m_1$ when the algorithm's signal is $a_1$ and $m_0$ when the algorithm's signal is $a_0$. Plugging in the first-best posterior beliefs we have:
\begin{align}\label{ica2}
   \sum_\omega  Pr(\omega | s_1, a_1, \theta_L)  \frac{\upsilon_H}{\upsilon_H + 1} &\geq  \sum_\omega  Pr(\omega | s_1, a_1, \theta_L), \\
       \sum_\omega  Pr(\omega | s_1, a_0, \theta_L)  \frac{1-\upsilon_H}{1-\upsilon_H + 1} &\geq  \sum_\omega  Pr(\omega | s_1, a_0, \theta_L).
\end{align}
Since, $\frac{\upsilon_H}{\upsilon_H + 1}$ and $\frac{1-\upsilon_H}{1-\upsilon_H + 1}$ are both less than one \eqref{ica2} is violated. Hence, the first-best is not incentive compatible for the low-type and cannot be sustained in equilibrium.\footnote{It can easily be shown the same is true for the high-type.} 
\end{proof}
Intuitively, under the first-best only the high-skill worker ever overrides the algorithm. Because of this, by overriding the algorithm the worker will be perceived as high-skill with probability one. Hence, both the low-skill and high-skill workers would have incentives to override the algorithm no matter what their signal is in order to be perceived as high-skill. As shown below, there exists an equilibrium in which the low-skill worker sometimes overrides the algorithm even though it is inefficient for him to do so. Before presenting the proposition, I establish the following definition regarding a ``informative equilibrium'', in which the distribution of the worker's message depends on the worker's type:\footnote{This definition follows closely the definition in \cite{guembel2009reputational}. The only difference is that it assumes that the manager interprets $m_1$ as $\omega_1$ is more likely and $m_0$ as $\omega_0$ is more likely. This is however, without loss of generality given that the messages have no intrinsic meaning to begin with.}
\begin{definition}
In a informative equilibrium, the manager's posterior belief that the worker is high-skill is always higher when he correctly forecasts the state, i.e., $\hat{\theta}( m_1, a, \omega_1) > \hat{\theta}( m_0, a, \omega_1)$  and $\hat{\theta}( m_0, a, \omega_0) > \hat{\theta}( m_1, a, \omega_0)$ for $a \in \{a_0, a_1\}$.
\end{definition}
This definition is useful because it rules out any babbling equilibria in which the the low and high-skill worker always follow the algorithm, override it, or randomize independently of their own signal.\footnote{Intuitively, the manager does not learn anything from the worker's forecast and subsequent state realization if the joint distribution of the worker's forecast and realization is independent of his skill.} 

Based on this definition, I next prove the following lemma which says the high-skill worker will always report his signal in any informative equilibrium.
\begin{lemma}\label{lemhighic}
    The high-skill worker always reports his signal in an informative equilibrium. 
\end{lemma}

\begin{proof}
See Appendix.
\end{proof}
The intuition of Lemma \ref{lemhighic} is straightforward. For a given set of the manager's posterior beliefs, it is always less costly for the high-skill worker to follow his own signal because his signal is more informative about the state than the low-skill worker's. Hence, in any informative equilibrium, the high-skill worker will always report his own signal. It will also be useful to show that the low skill worker never overrides the algorithm when he receives the same signal as it.
\begin{lemma}\label{lemlowic}
    The low-skill worker always reports his own signal when his signal equals that of the algorithm in an informative equilibrium, i.e., $m = m_0$ when $a = a_0$, $s = s_0$ and  $m = m_1$ when $a = a_1$, $s = s_1$. 
\end{lemma}

\begin{proof}
    See Appendix.
\end{proof}
The intuition for \eqref{lemlowic} is as follows: if the low-skill worker is going to override the algorithm, it is always less costly to do so when his signal is different than that of the algorithm. Indeed, as shown below, the low-skill worker will sometimes override the algorithm when his signal differs from it. 

Following Lemmas \ref{lemhighic} and \ref{lemlowic}, the following proposition characterizes the unique informative equilibrium. 

\begin{prop}[Algorithm Aversion]\label{equilibrium}
    There exists a unique, informative equilibrium in which the high-skill worker always reports his own signal and the low-skill worker reports the algorithm's signal with probability $\gamma \in (0,1)$ and his own signal with probability $1-\gamma$, where $\gamma$ is the unique solution to the following equation:
    \begin{gather}
   \sum_\omega  Pr(\omega | s_1, a_0, \theta_L) \hat{\theta}( m_0, a_0, \omega)- \sum_\omega  Pr(\omega | s_1, a_0, \theta_L)  \hat{\theta}( m_1, a_0, \omega)  =0, 
    \end{gather}
    where 
    \begin{align}
        \hat{\theta}( m_1, a_0, \omega_1) &= \frac{\upsilon _H}{\upsilon _H+(1-\gamma ) \upsilon _L}, \\
          \hat{\theta}( m_1, a_0, \omega_0) &=\frac{1-\upsilon _H}{1-\upsilon _H+(1-\gamma )
   \left(1-\upsilon _L\right)}, \\  
         \hat{\theta}( m_0, a_0, \omega_1) &= \frac{1-\upsilon _H}{\left(1-\upsilon
   _H\right)+\left(1-\upsilon _L\right)+\gamma \upsilon _L}, \\
          \hat{\theta}( m_0, a_0, \omega_0) &= \frac{\upsilon _H}{\upsilon _H+\upsilon _L+\gamma 
   \left(1-\upsilon _L\right)}.
    \end{align}
\end{prop}

\begin{proof}
Because of symmetry, we can again restrict focus to the case when the worker receives signal $s_1$. First consider the case when $a = a_1$. From Lemmas \ref{lemhighic} and \ref{lemlowic} both low and high-skill workers will report $m = m_1$. Now consider, the case in which $a = a_0$. The low-skill worker has the option to report his own signal or the algorithm's. Let $\gamma$ denote the probability that the low-skill worker reports the algorithm's signal, i.e., $m = m_0$. We can calculate the manager's posterior beliefs as follows:
\begin{align}
        \hat{\theta}( m_1, a_0, \omega_1) &= \frac{\upsilon _H}{\upsilon _H+(1-\gamma ) \upsilon _L}, \\
          \hat{\theta}( m_1, a_0, \omega_0) &=\frac{1-\upsilon _H}{1-\upsilon _H+(1-\gamma )
   \left(1-\upsilon _L\right)}, \\  
         \hat{\theta}( m_0, a_0, \omega_1) &= \frac{1-\upsilon _H}{\left(1-\upsilon
   _H\right)+\left(1-\upsilon _L\right)+\gamma \upsilon _L}, \\
          \hat{\theta}( m_0, a_0, \omega_0) &= \frac{\upsilon _H}{\upsilon _H+\upsilon _L+\gamma 
   \left(1-\upsilon _L\right)}.
    \end{align}
where the posteriors do not depend on $a$ because the algorithm's signal is independent of the worker's type $\theta$ and independent of the worker's signal $s$ conditional on the state $\omega$. First, it can easily be confirmed that these posterior beliefs satisfy the definition of an informative equilibrium, i.e., $\hat{\theta}( m_1, a_0, \omega_1) >  \hat{\theta}( m_0, a_0, \omega_1)$ and  $\hat{\theta}( m_0, a_0, \omega_0) >  \hat{\theta}( m_1, a_0, \omega_0)$. Next, define $G(\gamma)$ as the low-skill worker's difference in payoffs from reporting $m_1$ versus $m_0$ when he receives signal $s_1$ and the algorithm's signal is $a_0$: 
    \begin{gather}
    G(\gamma) \equiv \sum_\omega  Pr(\omega | s_1, a_0, \theta_L) \hat{\theta}( m_0, a_0, \omega) - \sum_\omega  Pr(\omega | s_1, a_0, \theta_L)  \hat{\theta}( m_1, a_0, \omega) .
    \end{gather}
First note from Proposition \ref{propfb}, when $\gamma = 1$, we are in the first-best use of information and $G(\gamma)$ is negative, implying that both the low-skill and high-skill worker will always report $m_1$. Hence $\gamma$ must be less than $1$ in any informative equilibrium. Differentiating $G(\gamma)$ w.r.t to $\gamma$ we have:
    \begin{gather}\label{gdec}
        G'(\gamma) = -\frac{\frac{(1-\alpha ) \upsilon _H \upsilon
   _L^2}{\left(\upsilon _H+(1-\gamma ) \upsilon
   _L\right){}^2}+\frac{\alpha  \left(1-\upsilon _H\right)
   \left(1-\upsilon _L\right){}^2}{\left(2-\gamma -\upsilon
   _H-(1-\gamma ) \upsilon _L\right){}^2}+\frac{\alpha 
   \upsilon _H \left(1-\upsilon _L\right){}^2}{\left(\gamma
   +\upsilon _H+(1-\gamma ) \upsilon
   _L\right){}^2}+\frac{(1-\alpha ) \left(1-\upsilon
   _H\right) \upsilon _L^2}{\left(2-\upsilon _H-(1-\gamma )
   \upsilon _L\right){}^2}}{\alpha - (2 \alpha -1)\upsilon_L}, 
    \end{gather}
which is negative because both the numerator and denominator are positive, implying that $G(\gamma)$ is decreasing in $\gamma$. When $\gamma = 0$ we have: 
    \begin{gather}
        G(0) = \frac{\left(\alpha -\upsilon _L\right) \left(\upsilon
   _H-\upsilon _L\right)}{\left(2-\upsilon _H-\upsilon
   _L\right) \left(\upsilon _H+\upsilon _L\right)
   \left(\alpha -(2 \alpha -1) \upsilon _L\right)},
    \end{gather}
 which is positive. Hence, there exists a unique $\gamma$ such that $G(\gamma) = 0$. Note that whenever $a = a_0$, the low-skill worker must be indifferent between reporting $m_0$ and $m_1$. Hence, this is the unique informative equilibrium.
\end{proof}

Proposition \ref{equilibrium} shows that, despite the algorithm providing a more informative signal than the low-skill worker's, the low-skill worker sometimes overrides it when his own signal is different than that of the algorithm's. Intuitively, the low-skill worker cannot always follow the algorithm, otherwise overriding the algorithm signals with probability one that he is high-skill. In other words, the worker must override the algorithm occasionally in order to convince their manager they are not dispensable. Hence, Proposition \ref{equilibrium} shows how completely rational reputational concerns can endogenously create algorithm aversion.

\cite{agarwal2023combining} show that radiologists are less likely to follow an algorithm when they believe that the algorithm's forecast is more uncertain. Motivated by this empirical finding, I next show that the more uncertain the algorithm's signal is, the more likely the low-skill worker overrides the algorithm in the unique informative equilibrium.  
\begin{prop}\label{prop:unc}
 Algorithm aversion increases in the uncertainty of the algorithm, i.e., $\gamma$ is increasing in $\alpha$ in the informative equilibrium. 
\end{prop}
\begin{proof}
Applying the implicit function theorem we have: 
   \begin{gather}
       \frac{ d \gamma }{d \alpha} = -\frac{ \partial G(\gamma)/ \partial \alpha}{\partial G(\gamma)/\partial \gamma},
    \end{gather}
    where the denominator is negative from \eqref{gdec} and the numerator is equal to:
    \begin{align}
     & \sum_\omega  \frac{\partial Pr(\omega | s_1, a_0, \theta_L)}{\partial \alpha} \hat{\theta}( m_0, a_0, \omega) - \sum_\omega  \frac{\partial Pr(\omega | s_1, a_0, \theta_L)}{\partial \alpha}  \hat{\theta}( m_1, a_0, \omega) \\ 
     &= \frac{\left(1-\upsilon _L\right) \upsilon _L}{\left(\alpha +
   (2 \alpha -1) \upsilon _L\right){}^2} \left( \hat{\theta}( m_0, a_0, \omega_0) - \hat{\theta}( m_0, a_0, \omega_1) + \hat{\theta}( m_1, a_0, \omega_1) - \hat{\theta}( m_1, a_0, \omega_0) \right), 
    \end{align}
 where first term is positive and the second term is also positive based on the definition of an informative equilibrium. Hence, $\frac{ d \gamma }{d \alpha}$ is positive and $\gamma$ is increasing in $\alpha$ in the informative equilibrium.

\end{proof}
This result may at first glance seem obvious. If the algorithm is less accurate, the low-skill worker should be less likely to follow it. However, from the manager's perspective so long as Assumption \ref{a1} is still satisfied, it is always efficient for the low-skill worker to report the algorithm's signal. Moreover, there are two effects here. On the one hand, it is less costly to override the algorithm given that the algorithm's signal is uninformative. On the other hand, the high-skill worker also becomes more likely to override the algorithm which increases the signaling value of overriding the algorithm for the low-skill worker. To see this, we can calculate the probability the high-skill worker receives a different signal than the algorithm:
\begin{align}
    Pr(s_1, a_0 | \theta_H) &= \Pr(s_1, a_0 | \omega_1, \theta_H) \Pr(\omega_1 | \theta_H) + \Pr(s_1, a_0 | \omega_0, \theta_H) \Pr(\omega_0 | \theta_H) \\ 
    &= \frac{1}{2}\left(\Pr(s_1, a_0 | \omega_1, \theta_H) + \Pr(s_1, a_0 | \omega_0, \theta_H) \right) \\ 
        &= \frac{1}{2}\left( (1-\alpha)\upsilon_H  + \alpha (1-\upsilon_H) \right). \label{probskill}
\end{align}
Differentiating the RHS of \eqref{probskill} with respect to $\alpha$ we have $\frac{1}{2} - \upsilon_H$, which is negative. Hence, $Pr(s_1, a_0 | \theta_H)$ is decreasing in $\alpha$, which implies that the less accurate the algorithm's signal is, the more likely the high-skill worker receives the opposite signal of the algorithm. Since in equilibrium the high-skill worker always overrides the algorithm whenever he receives the opposite signal of the algorithm, this result implies that the high-skill worker more frequently overrides the algorithm. In turn, this increases the signaling content of overriding the algorithm in equilibrium which further incentivizes the low-skill worker to override the algorithm.

\section{Labor Market Implications}

In this section I explore some of the labor market implications of reputational algorithm aversion.

Several studies find that even when being able to incorporate an algorithm's forecast, humans' forecasts often perform worse than the algorithm on its own (e.g., \cite{agarwal2023combining} and \cite{angelova2023algorithmic}). The next proposition shows that this phenomenon can also occur in the model.
\begin{prop}
 When the average accuracy of the worker's signal is less than that of the algorithm's signal, i.e., $\frac{1}{2}(\upsilon_L + \upsilon_H) > \alpha$, the expected accuracy of the worker's forecast can be lower than that of the algorithm's. 
\end{prop}

\begin{proof}
Again, because of symmetry we can restrict focus to the case in which the worker reports $m_1$ and examine how likely $\omega = \omega_1$ as compared to when the algorithm produces signal $a_1$.\footnote{Because of symmetry $m_1$ and $a_1$ are always realized half the time.} For the worker to be more accurate than the algorithm, the following condition must hold:  
\begin{align}
    Pr(\omega_1 | m_1) &\geq Pr(\omega_1 | a_1) ,
    \implies Pr(\omega_1 | m_1) \geq \alpha. \label{acc1}
\end{align}
We can rewrite the $Pr(\omega_1 | m_1)$ as follows: 
\begin{align}
    &= Pr( m_1 | \omega_1, a_0 ) Pr(a_0|\omega_1) +  Pr( m_1 | \omega_1, a_1 ) Pr(a_1|\omega_1) \\
&= \frac{1}{2} \left(\alpha  \gamma +\upsilon _H+(1-\gamma ) \upsilon _L\right) . \label{eq:acc}
\end{align}
From this expression notice that if the average accuracy of the worker is higher than that of the algorithm, i.e., $\frac{1}{2}(\upsilon_L + \upsilon_H) > \alpha$, then the worker's forecast is always more accurate than the algorithm. However, when $\frac{1}{2}(\upsilon_L + \upsilon_H) < \alpha$ this may no longer be the case. For example suppose that $\upsilon_L = 0.55$, $\alpha = 0.60$ and $\upsilon_H = 0.62$. Solving the informative equilibrium numerically, we have that $\gamma = 0.0148$. Plugging these values into the expression for the worker's forecast accuracy, \eqref{eq:acc}, we have $0.58537$ which is clearly less than the algorithm alone $(0.60)$. Hence, there are cases in which the algorithm performs better than the worker's forecast even though the worker has access to the algorithm's signal. 
\end{proof}

This result suggests that despite the potential of humans improving the forecast of the algorithm, their reputational concerns can make it such that their forecast is worse than the algorithm's on its own. Hence, to the extent wages are determined by workers' ability to improve the algorithm's forecast, reputational algorithm aversion should reduce workers' wages relative to the wages they would receive if they used information efficiently. Moreover, if there are costs of training or maintaining workers, this phenomenon could cause more firms to avoid hiring them at all. 

Next I explore how algorithm aversion affects the relationship between the worker's wage and the algorithm's accuracy. I do not explicitly model the wage setting process, but instead simply assume wages are increasing in the difference in the expected accuracy of the worker's forecast and that of the algorithm.

First, it will be useful to subtract the algorithm's accuracy, $\alpha$, from the human's forecast accuracy given in \eqref{eq:acc}:
\begin{gather}\label{eq:diffacc}
\frac{1}{2} \left(\alpha  \gamma +\upsilon _H+(1-\gamma ) \upsilon _L\right)  - \alpha. 
\end{gather}
Differentiating \eqref{eq:diffacc} with respect to $\alpha$ we have:
\begin{gather}\label{acc:der}
    \frac{1}{2}\left(\gamma + (\alpha - \upsilon_L ) \frac{d \gamma}{d \alpha} - 2\right).
\end{gather}
Notice that increasing the accuracy of the algorithm has two effects on the worker's marginal contribution to the forecast's accuracy. The obvious direct effect is that a more accurate algorithm makes it such that the human becomes less valuable for forecasting the outcome. However, from Proposition \ref{prop:unc}, $\frac{d \gamma}{d \alpha}$ is positive, which creates a counteracting attenuating effect due to the worker endogenously following the algorithm more frequently when the algorithm becomes more accurate. In other words, the human's added value increases due to a reduction in algorithm aversion as the algorithm becomes more accurate. Hence, while overall wages should be overall lower, algorithm aversion also flattens the relationship between wages and algorithm accuracy. Hence, relative to a benchmark in which workers use information efficiently, this result suggests that algorithms' effect on workers' wages should be more concentrated in the extensive margin of algorithm adoption rather than the intensive margin of improvements in algorithm accuracy.\footnote{Of course, this assumes algorithms are sufficiently accurate to be useful at all to begin with.}

It is worth also briefly discussing how algorithm aversion affects the adoption of algorithms. Consider a situation in which the manager has a worker in place and can pay a fixed-cost to adopt the algorithm. Here it will be useful to consider the difference between in accuracy between the worker's forecast with the algorithm and that without it: 
\begin{gather}\label{endogenous}
    \frac{1}{2} \left(\alpha  \gamma +\upsilon _H+(1-\gamma ) \upsilon _L\right) - \frac{1}{2}\left(\upsilon_L + \upsilon_H\right) = \frac{1}{2}\left(\alpha - \upsilon_L\right) \gamma 
\end{gather}
Intuitively, \eqref{endogenous} says that the value of the algorithm is equal to the difference in signal accuracy of the low-skill worker and the algorithm, times the probability that the worker is low-skill and follows the algorithm. Naturally, the less often the worker follows the algorithm, i.e., the lower $\gamma$, the less valuable the algorithm is to the firm, whereas if $\gamma = 1$ then we are in the first-best use of information. Hence, to the extent that workers cannot simply be replaced by algorithms the reputational concerns of the worker can also deter the firm from adopting the algorithm to begin with. 

\section{Model Discussion and Applications}

In this section I discuss some of the key assumptions in the model as well as its applications. 

First, one may wonder whether the problem can be resolved if the manager simply asks the worker to report his own forecast after which the manager makes his own forecast. However, this ultimately collapses to the original problem. The low-skill worker would still have incentives to at least sometimes report the his own signal rather than the algorithm's. Relatedly, in some contexts it may not be possible for the worker to convey his signal to the manager, but must instead  make a forecast or take an action himself that indirectly conveys that information.\footnote{This could occur if the information cannot be substantiated or if the worker is making a decision, such as an investment, that must be implemented by himself, not the manager.} 

One potential solution is to make the algorithm's signal only observable to the worker. In this case, the worker's forecast cannot be compared to the algorithm's and hence the worker's forecast cannot convey information to the manager on its own. There are two practical problems with this. First, the algorithm is likely developed by other people in the firm (or outside the firm) with expertise in artificial intelligence.\footnote{This is the main motivation for making the algorithm's output observable to both the worker and the manager - the fact that those who develop algorithms are likely to be different than those who implement decisions using them.} For example, data scientists within the firm can use firm-level inputs to forecast the same outcomes that other employees are tasked with forecasting using their own skill or intuition. The data scientists maintaining the algorithm will inevitably want to see the information produced by it in order to evaluate and continue developing it. Moreover, even if the firm could make the algorithm's signal only observable to the worker, it may not be optimal for the firm because it prevents them from learning about the skill of that worker.

The manager could also consider hiding or delaying the availability of the algorithm's signal to the worker until after the worker reports his forecast. However, here the worker's problem would collapse to the benchmark case without the algorithm, which can easily be shown has a lower forecast accuracy than the informative equilibrium with the algorithm. Moreover, even if the manager ultimately takes responsibility for the forecast or action, the manager cannot distinguish the skill of the worker given that both low and high-skill workers will always report their information. Hence, the manager either will always follow the worker or the algorithm, depending on which has a higher average accuracy. 

In the informative equilibrium the low-skill worker plays a mixed strategy in which he is indifferent between reporting the algorithm's signal or his own. While this may not seem realistic at first, mixed strategies can also be interpreted as pure strategies with random disturbances (e.g., \cite{harsanyi1973games}). Equivalently, I could assume that the low-skill worker's signal precision is a continuous random variable with mean $\upsilon_L$ in which the worker reports his own signal whenever the signal precision is above some threshold. This simple alteration would lead to essentially the same equilibrium in pure strategies.

It is important to emphasize that the signaling mechanism in this model is distinct from that in costly signaling settings (e.g., \cite{spence1978job}). In costly signaling, the signal is exogenously costly (e.g., going to college), while in this model there is no intrinsic cost to overriding the algorithm.\footnote{This is why equilibrium refinements have less bite in cheap-talk models as compared to signaling models (e.g., \cite{chen2008selecting}) because it is costless for senders to randomize over actions on the equilibrium path.} Rather, the forecast conveys information about the worker's skill purely because the high-skill worker has more confidence to override the algorithm given that he receives a more precise signal than the low-skill worker.

One may also wonder whether the main mechanism of the model is an artifact of the binary structure of signals and outcomes. If outcomes are continuous then the problem is essentially the same because the worker still has only two potential pieces of information he can report. If the signal is continuous the mechanism would still be present; however, rather than the override decision being a signal, the distance of the worker's forecast from the algorithm would be a signal of the worker's quality.\footnote{This logic applies whether or not the outcome is continuous or not.} Intuitively, the high-skill worker would have more confidence to report a forecast farther away from the algorithm given that his signal is more informative than the low-skill worker.

As is typical in career concerns models (e.g., \cite{holmstrom1999managerial}), I assume there are no long-term contracts. However, it can easily be shown that if the firm can write contracts, then the first-best use of information can be achieved. Specifically, the firm could provide a long-term wage that is invariant of the worker's performance. In this case, the worker would have no reputational concerns and hence, would be willing to report in accordance with the first-best use of information.\footnote{This case is knife-edge because the worker is indifferent between the two messages; however, while I do not directly model firm-value, the firm could also provide the worker with some equity to further incentivize them.} This is unlikely in practice for several reasons. First, long-term contracts may not be enforceable.\footnote{In particular due to at-will employment the high-skill worker can potentially break the contract after one period for a higher offer outside the firm.} Second, although I do not explicitly model the hiring process, offering these types of long-term contracts could lead to a lemons problem in which only low-skill workers join the firm. Finally, and perhaps largely due to these aforementioned issues, these types of long-term contracts are not often seen in practice.

In terms of applications, it is important that the decisions/forecasts the worker makes are consequential enough and not too frequent. Otherwise, the firm can learn quickly the worker's skill and the inefficiency will be small. For example, a natural application is an employee in a firm's treasury department projecting the future cash flows of the firm or making an investment decision.

It is also important to emphasize that the model applies to situations in which following an algorithm has a reputational cost for a human. In practice, there are some instances in which a human is presented with a forecast either produced by a human or an algorithm and the human is reluctant to follow the algorithm even though there is no obvious reputational cost of doing so (e.g., see \cite{greig2023algorithm} in the case of robo-advising). The model best applies to situations in which an expert has to make some type of forecast or decision based on his own information and an algorithm's and he is ultimately evaluated by others. For example, the model applies to any situation in which an employee of a firm must incorporate information from an algorithm into his decisions (e.g., cash flow forecasting, investment decisions and financial policies) and those decisions are ultimately evaluated by his manager or the firm as a whole.

Finally, it is worth mentioning that while I assume the public signal is generated by an algorithm, the model would also apply to any public information observable to both the manager and the worker. Nonetheless, algorithms are a particularly relevant and important application given that they produce information that is 1) in many cases becoming more accurate than humans and 2) observable within and across firms and institutions.\footnote{As discussed above, very often the employees who develop algorithms are different than those asked to use them, e.g., the treasurer of a firm. Hence, it is unlikely that the user of the algorithm is the only one who can observe its output.}

\section{Conclusion}
Rapid advances in artificial intelligence have raised many questions about the future role of human workers in society. While for certain tasks, algorithms may displace humans entirely, for many others tasks, humans will have to work alongside them. This paper shows how humans' concerns for their own reputations can impede human-AI collaboration. Completely rational humans may engage in ``algorithm aversion'' and override the algorithm even if it is inefficient to do so. The intuition for this idea is straightforward: humans have to sometimes override algorithms in order to prove they are not replaceable. 

Practically, this can lead to one of two problems: either firms become slower to adopt algorithms, or firms forego using humans for a wider variety of tasks than would be necessary in the absence of humans' reputational concerns.  

The solution to reputational algorithm aversion is not obvious. As mentioned earlier, reputational concerns can be alleviated if the decision-makers are also those put in charge of developing the algorithm (e.g., \cite{dietvorst2018overcoming}). However, in many cases, these decision-makers do not have expertise in AI systems. For example, most treasurers of firms likely have little understanding of how a deep learning algorithm works. Similarly, only allowing the decision-maker to see the algorithm's signal would also mitigate reputational concerns. However, this would obviously be difficult if the algorithm is developed by other employees of the firm. Moreover, after promising to keep the algorithm's forecast private, firms may be tempted ex-post to observe the algorithm's signal in an attempt to learn about workers' skills. 

Taken together, this paper identifies an entirely rational, yet realistic force that can cause humans to be averse to incorporating information generated by algorithms into their own decisions. Moreover, as AI systems become more and more integrated into firms, I argue that analyzing human-AI interactions through the lens of strategic communication games will be extremely fruitful method to understand these interactions.

\clearpage 

\bibliography{bib.bib}
	\bibliographystyle{jfnew}

 \clearpage 
 
\appendix

\setcounter{equation}{0}
\renewcommand{\theequation}{\thesection\arabic{equation}}

\section{Additional Proofs}

\begin{proof}[\normalfont \textbf{Proof of Lemma \ref{lemhighic}}]
First, it will be useful to prove the following claim that the high-skill worker never mixes in any informative equilibrium. 

\begin{claim}
 The high-skill worker plays a pure strategy in any informative equilibrium.
\end{claim}

\begin{proof}
   Because of symmetry, it is without loss of generality to consider the case in which the algorithm's signal is $a_1$. First consider the case when $s = s_1$. Suppose to the contrary that the high-skill worker reports $m_1$ with some probability $p \in (0,1)$. For this to be an equilibrium, he must be indifferent between sending $m_0$ and $m_1$: 
    \begin{align}\label{ichigh2m}
      &Pr(\omega_1 | s_1, a_1, \theta_H) \left(\hat{\theta}( m_1, a_1, \omega_1) - \hat{\theta}( m_0, a_1, \omega_1)\right) \\ &= Pr(\omega_0 | s_1, a_1, \theta_H) \left(\hat{\theta}(m_0, a_1, \omega_0) - \hat{\theta}( m_1, a_1, \omega_0)\right). 
\end{align}
If the high-skill worker mixes when $s = s_1$, he cannot also mix when $s = s_0$. To see this, we solve for $\hat{\theta}( m_1, a_1, \omega_1)$ in \eqref{ichigh2m} and plug the solution into the high-skill worker's indifference condition when $s= s_0$:
\begin{gather}\label{eqeq}
    \frac{(1-\alpha) \left(2 \upsilon _H-1\right)
   (\hat{\theta}(m_1,a_1,\omega_0) - \hat{\theta}(m_0,a_1,\omega_0)) }{\upsilon _H \left((2 \alpha -1)
   \upsilon _H -\alpha\right)} = 0.
\end{gather}
However, for \eqref{eqeq} to hold it requires that the  $\hat{\theta}(m_1,a_1,\omega_0) = \hat{\theta}(m_0,a_1,\omega_0)$, which contradicts the equilibrium being informative.

Now consider the case in which the high-skill worker mixes when we he receives signal $s_1$ and plays a pure strategy when he receives signal $s_0$. For the high-skill worker to report $m_1$ when he receives signal $s_0$, the LHS of \eqref{eqeq} must be positive. However, this can only be the case when $\hat{\theta}(m_1,a_1,\omega_0) - \hat{\theta}(m_0,a_1,\omega_0)$, which contradicts the equilibrium being informative. Hence, the only other possibility is that the high-skill worker reports $m_0$ when he receives signal $s_0$. If we plug in the solution for $\hat{\theta}( m_1, a_1, \omega_1)$ obtained from \eqref{ichigh2m} into the low-skill worker's difference in payoff from reporting $m_1$ versus $m_0$ after receiving signal $s_0$ we have: 
   \begin{align}\label{eqeq2}
     \frac{(1-\alpha)(\upsilon_H - \upsilon_L)\left(\hat{\theta}( m_1, a_1, \omega_0) -\hat{\theta}(m_0, a_1, \omega_0)\right)}{\upsilon_H \left(1 - \alpha +(2\alpha - 1)\upsilon_L\right)}.
\end{align}
For the equilibrium to be informative it must be that $\hat{\theta}( m_1, a_1, \omega_0) < \hat{\theta}(m_0, a_1, \omega_0)$, which implies \eqref{eqeq2} is negative, which then implies the low-skill worker would also report $m_0$ when $s = s_0$. The indifference condition of the high-skill worker, \eqref{ichigh2m}, implies that the low-skill worker also always reports $m_0$ when $s = s_1$ because $Pr(\omega_1 | s_1, a_1, \theta_H) > Pr(\omega_1 | s_1, a_1, \theta_L)$ and $Pr(\omega_0 | s_1, a_1, \theta_L) > Pr(\omega_0 | s_1, a_1, \theta_H)$. Hence, under this scenario the low-skill worker always reports $m_0$ while the high-skill worker sometimes reports $m_1$. However, based on the same logic as in Proposition \ref{propfb}, this can never be an equilibrium because the worker will instantly be identified as high-skill if he reports $m_1$ when $a = a_1$. This can easily be confirmed by computing the posteriors and plugging them into either the low or high-skill worker's incentive compatibility constraints.

Since the signals are symmetric this also proves that the high-skill worker mixing when he receives $s_0$ can also never be an equilibrium. The same logic steps apply when the algorithm's signal is $a_0$ instead of $a_1$. Hence, the high-skill worker always plays a pure strategy in any informative equilibrium. 
\end{proof}

Given that the high-skill worker does not mix, there are four possible cases to consider:
\begin{enumerate}[label=\textbf{Case \arabic*:}, left=0pt, itemindent=20pt]
    \item The high-skill worker always reports the opposite of his signal
    \item The high-skill worker always reports the algorithm's signal
    \item The high-skill worker always reports the opposite of the algorithm's signal
    \item The high-skill worker always reports his signal 
\end{enumerate}

Given the symmetry of the problem, it is without loss of generality to assume $a = a_1$ for each case.

First consider Case 1 in which the high-skill worker always reports the opposite of his signal. This requires the following inequalities to hold:  
\begin{align}
   \sum_\omega  Pr(\omega | s_1, a_1, \theta_H)  \hat{\theta}( m_1, a_1, \omega) & \leq  \sum_\omega  Pr(\omega | s_1, a_1, \theta_H)  \hat{\theta}( m_0, a_1, \omega) \\
   \sum_\omega  Pr(\omega | s_0, a_1, \theta_H)  \hat{\theta}( m_0, a_1, \omega) & \leq  \sum_\omega  Pr(\omega | s_0, a_1, \theta_H)  \hat{\theta}( m_1, a_1, \omega).
 \end{align} 
 The first inequality can be written as:
 \begin{align}
     Pr(\omega_1 | s_1, a_1, \theta_H) \left(\hat{\theta}( m_1, a_1, \omega_1) - \hat{\theta}( m_0, a_1, \omega_1)\right) \leq \\
        Pr(\omega_0 | s_1, a_1, \theta_H) \left(\hat{\theta}( m_0, a_1, \omega_0) - \hat{\theta}( m_1, a_1, \omega_0)\right),
 \end{align}
and the second to:
 \begin{align}
     Pr(\omega_0 | s_0, a_1, \theta_H) \left(\hat{\theta}( m_0, a_1, \omega_0) - \hat{\theta}( m_1, a_1, \omega_0)\right) \leq \\
        Pr(\omega_1 | s_0, a_1, \theta_H) \left(\hat{\theta}( m_1, a_1, \omega_1) - \hat{\theta}( m_0, a_1, \omega_1)\right).
 \end{align}
 Both sides of these inequalities must be positive in an informative equilibrium. Moreover, they imply that the low-skill worker would also report the opposite of his own signal because $P(\omega_0|s_1,a_1,\theta_L) > P(\omega_0|s_1,a_1,\theta_H)$  and  $P(\omega_1|s_0,a_1,\theta_L) > P(\omega_0|s_0,a_1,\theta_H)$. However, if the low-skill worker also reports the opposite of his signal, after calculating the corresponding posterior beliefs of the manager, it immediately becomes clear that the equilibrium is not informative:
 \begin{gather}
     \hat{\theta}( m_1, a, \omega_1) = \frac{1-\upsilon_H}{2-\upsilon_H -\upsilon_L} < \frac{\upsilon_H}{\upsilon_H + \upsilon_L} = \hat{\theta}( m_0, a, \omega_1) \quad a \in \{a_0, a_1\}. 
 \end{gather}
Hence, the high-skill worker cannot always report the opposite of his signal. 

Cases 2 and 3 can easily be ruled out based on the following logic. In order for Case 2 or 3 to be an equilibrium, the low-skill worker must also always either report the algorithm's signal (Case 2) or the opposite of the algorithm's signal (Case 3), otherwise the low-skill worker would be identified with probability one as being low-skill. However, in either of these cases the equilibrium would not be informative, which leads to a contradiction. Hence, the only remaining possible case is Case 4 in which the high-skill worker reports his signal.

\end{proof}

\begin{proof}[\normalfont \textbf{Proof of Lemma \ref{lemlowic}}]
Because of the symmetry of the problem, it is again without loss of generality to consider the case in which the algorithm's signal is $a_1$. When $s = s_1$ suppose that the low-skill worker reports $m_1$ with some probability $p \in (0,1)$. Then the following condition must hold:
\begin{align}\label{iclow}
          &Pr(\omega_1 | s_1, a_1, \theta_L) \left(\hat{\theta}( m_1, a_1, \omega_1) - \hat{\theta}( m_0, a_1, \omega_1)\right) \\ 
     &= Pr(\omega_0 | s_1, a_1, \theta_L) \left(\hat{\theta}(m_0, a_1, \omega_0) - \hat{\theta}( m_1, a_1, \omega_0)\right).
\end{align}
If we solve for $\hat{\theta}( m_1, a_1, \omega_1)$ in \eqref{iclow} and plug this into the difference in the low-skill worker's payoff from reporting $m_1$ versus $m_0$ when $s=s_0$ we have:
\begin{align}\label{iclow2}
    \frac{ (1-\alpha ) \left(2 \upsilon
   _L-1\right)\left(\hat{\theta}(m_1,a_1,\omega_0) - \hat{\theta}(m_0,a_1,\omega_0)\right)}{\upsilon _L \left(\alpha -(2 \alpha -1) \upsilon _L\right)}.
\end{align}
If the equilibrium is informative, i.e., $\hat{\theta}(m_1,a_1,\omega_0) < \hat{\theta}(m_0,a_1,\omega_0)$, then \eqref{iclow2} is negative, which implies that the low-skill worker always reports $m_0$ when he his signal is $s_0$. We can then calculate the posteriors and plug them into \eqref{iclow}:
\begin{align}
    &Pr(\omega_1|s_1,a_1) \left(\frac{\upsilon _H}{\upsilon _H+p \upsilon _L} -\frac{1-\upsilon _H}{2-\upsilon _H-p \upsilon _L} \right) \\    
    &= Pr(\omega_0|s_1,a_1) \left(\frac{\upsilon _H}{1+\upsilon _H-p \left(1-\upsilon _L\right)} -  \frac{1-\upsilon _H}{1-\upsilon _H-p \left(1-\upsilon _L\right)}\right).
\end{align}
Further simplifying we have:
\begin{align}\label{iclow3}
    \frac{\frac{(1-\alpha ) \left(1-\upsilon _H-p\left(1-\upsilon
   _L\right)\right) \left(1-\upsilon _L\right)}{\left(1-\upsilon
   _H+p \left(1-\upsilon _L\right)\right) \left(1+\upsilon
   _H-p\left(1-\upsilon _L\right)\right)}+\frac{\alpha  \upsilon
   _L \left(\upsilon _H-p\upsilon _L\right)}{\left(2-\upsilon
   _H-p\upsilon _L\right) \left(\upsilon _H+p\upsilon
   _L\right)}}{1-\alpha +(2 \alpha -1) \upsilon _L} = 0.
\end{align}
However, the LHS of \eqref{iclow3} is strictly positive which means it is not incentive compatible for the low-skill worker to mix when $s=s_1$. 

Now consider the low-skill worker reporting $m_0$ when $s=s_1$. Let $p$ denote the probability the low-skill worker reports $m_1$ when $s= s_0$. We can then calculate the posteriors and plug them into the low-skill worker's incentive compatibility constraint when $s = s_1$:
\begin{align}
    &Pr(\omega_1|s_1,a_1) \left(\frac{\upsilon _H}{\upsilon _H+p \left(1-\upsilon _L\right)} -\frac{1-\upsilon _H}{2-\upsilon _H-p \left(1-\upsilon _L\right)} \right) \\    
    &\leq  Pr(\omega_0|s_1,a_1) \left(\frac{\upsilon _H}{1+\upsilon _H-p \upsilon _L} -  \frac{1-\upsilon _H}{1-\upsilon _H+p \upsilon _L}\right).
\end{align}
Further simplifying we have: 
\begin{align}\label{iclow4}
    \frac{\frac{\alpha  \upsilon _H \upsilon _L}{\upsilon _H+p
   \left(1-\upsilon _L\right)}-\frac{\alpha  \left(1-\upsilon _H\right)
   \upsilon _L}{2-\upsilon _H-p \left(1-\upsilon
   _L\right)}+\frac{(1-\alpha ) \left(1-\upsilon _H\right) \left(1-\upsilon
   _L\right)}{1-\upsilon _H+p \upsilon _L}-\frac{(1-\alpha )
   \upsilon _H \left(1-\upsilon _L\right)}{1+\upsilon _H-p
   \upsilon _L}}{1-\alpha +(2 \alpha -1) \upsilon _L} \leq 0.
\end{align}
However, the LHS of \eqref{iclow4} is strictly positive which contradicts the low-skill worker reporting $m_0$ when $s = s_1$. To see that \eqref{iclow4} is positive, we can differentiate the LHS of \eqref{iclow4} with respect to $p$ and we have:
\begin{align}
-\frac{\left(1-\upsilon _L\right) \upsilon _L \left(\frac{\alpha  \upsilon
   _H}{\left(\upsilon _H+p\left(1-\upsilon
   _L\right)\right){}^2}+\frac{\alpha  \left(1-\upsilon
   _H\right)}{\left(2-\upsilon _H-p\left(1-\upsilon
   _L\right)\right){}^2}+\frac{(1-\alpha ) \upsilon _H}{\left(1+\upsilon
   _H-p\upsilon _L\right){}^2}+\frac{(1-\alpha )
   \left(1-\upsilon _H\right)}{\left(1-\upsilon _H+p\upsilon
   _L\right){}^2}\right)}{1-\alpha +(2 \alpha -1) \upsilon _L} < 0
\end{align}
Hence, \eqref{iclow4} is decreasing in $p$. However, when $p = 1$ \eqref{iclow4} simplifes to:
\begin{align}
    \frac{\left(\alpha +\upsilon _L-1\right) \left(\upsilon _H+\upsilon
   _L-1\right)}{\left(1-\left(\upsilon _H-\upsilon _L\right){}^2\right)
   \left(1-\alpha +(2 \alpha -1) \upsilon _L\right)},
\end{align}
which is strictly positive. Hence, \eqref{iclow4} is violated. Therefore, the low-skill worker reports $m_1$ when $s = s_1$ and $a = a_1$. Because of symmetry, the same steps can be used to show the low-skill worker reports $m_0$ when $s= s_0$ and $a = a_0$. 
\end{proof}

\end{document}